\documentclass{article}

\usepackage[utf8]{inputenc}
\usepackage{amssymb}
\usepackage{graphicx}
\usepackage{amsmath}
\usepackage{enumerate}
\usepackage{hyperref}

\newcommand{\N}{\mathbb{N}}
\newenvironment{proof}{{\em Proof:} } {\hfill $\Box$}

\newtheorem{theorem}{Theorem}

\def\bbN{\mathbb{N}}
\def\vX{X}
\def\vY{Y}

\def\vless{\leq}
\def\nvless{\not\leq}
\def\vnext{\rightarrow}

\begin{document}

\title{Non-dominating sequences of vectors using only resets and increments}

\author{Wojciech Czerwi\'nski,
Tomasz Gogacz,
Eryk Kopczy\'nski }

\maketitle

\begin{abstract}
We consider sequences of vectors from $\N^d$.
Each coordinate of a vector can be reset or incremented by $1$ with respect to the same
coordinate of the preceding vector.
We give an example of non-dominating sequence, like in Dickson's Lemma,
of length $2^{2^{\theta (n)}}$, what matches the previously known upper bound.
\end{abstract}

\section{Introduction}

In this paper we consider $d$-dimensional vectors of non-negative integers $\N$.
A vector $v$ \emph{dominates} vector $w$ if on each coordinate of $v$ is bigger or equal than $w$.
Sequence of vectors $v_1, v_2, \ldots$ is \emph{non-dominating} if there is no pair $v_i$, $v_j$
with $i < j$ such that $v_j$ dominates $v_i$.

Dickson has shown in~\cite{Dickson} that there is no infinite non-dominating sequence of vectors.
However, they can be arbitrarily long even when starting from a fixed vector, for example:
$(1, 0), (0, n), (0, n-1), \ldots, (0, 1), (0,0)$.

Figueira et al.~\cite{DBLP:conf/lics/FigueiraFSS11} and McAllon~\cite{DBLP:journals/tcs/McAloon84} have considered
sequences of vectors from $\N^d$ which fulfill some given restriction.
A possible restriction can be described by a function $f : \N \to \N$ such that for all $i$,
all coordinates of the $i$-th vector in the sequence do not exceed $f(i)$.
For different families of functions authors deliver different maximal lengths of non-dominating sequences,
however even for very restricted sequences bounds are non primitive recursive.

A problem stated in~\cite{DBLP:conf/wia/Zimmermann09} and~\cite{DBLP:conf/atva/HornTW08} concerns
games, but it can be formulated in terms of such sequences. 
Authors consider a game on a finite graph with special vertices - request and response vertices. 
The following condition is added to the objective of a game:
whenever a request vertex of type $k$ is visited, an response vertex of type $k$ has to be reached in the future of the play.

A \emph{summary} in such a play consists of:
\begin{itemize}
 \item(1) a position in the graph
 \item(2) waiting time since first unanswered request for each type of (vector part). 
\end{itemize}

Authors of~\cite{DBLP:conf/wia/Zimmermann09} and~\cite{DBLP:conf/atva/HornTW08} want to determine which player has a winning strategy.
They have shown that there is a winning strategy which depends only on the summary.
Moreover only strategies for which vector part is a non-dominating sequence can be considered.
Hence, an upper bound on length of such sequences would give an upper bound on the complexity of the mentioned game.
It motivates investigations of an effective variant of Dickson's Lemma where for each vector
each coordinate can be either
\begin{itemize}
 \item reset to $0$ (as a result of visiting a response vertex); or
 \item increased by $1$ (waiting phase)
\end{itemize}
with respect to the same coordinate in the preceding vector.

So far only single exponential lower bound and doubly exponential upper bound were known.
This paper is devoted to present and explain a doubly exponential lower bound.

\section{Notation and definitions}

As mentioned above we consider $d$-dimensional vectors.
In each step, every coordinate can be either reset to 0 or incremented.
For two coordinate values $a, b \in \bbN$, we say that $a \vnext b$ iff
either $b=a+1$ or $b=0$.

By $v^\ell_i \in \bbN$ we denote the value of the $\ell$-th coordinate after $i$ steps.
By $v^\ell$ we denote the sequence of values (i.e., the history) of the $\ell$-th coordinate, and
by $v_i \in \bbN^d$ we denote the $i$-th vector in the sequence (state of the whole
system after $i$ steps).

For two vectors $v,w \in \bbN^d$, we say that $v \vnext w$ iff $v^\ell \vnext w^\ell$ for each $\ell$.
We say that a sequence of vectors of length $n$ is \emph{valid}
iff $v_i \vnext v_{i+1}$ for each $i < n-1$. Moreover, we say that a sequence of
vectors is \emph{cyclic} iff it is valid and $v_{n-1} \vnext v_0$.

For two vectors $v,w \in \bbN^d$, we write $v \vless w$ iff for each coordinate $\ell$ we have $v^\ell \vless w^\ell$.
Otherwise, we write $v \nvless w$.
Thus a valid sequence $(v_i)$ is non-dominating iff whenever $i < j$ we have $v_i \nvless v_j$.

Let $L_d$ be the maximum length of a non-dominating valid sequence for the given $d$.

In this paper, we show the double exponential lower bound for $L_d$.

\section{Main Idea}
In this section we present the main idea behind our solution.
Let us first consider an easier variant, where we allow coordinates to stay unchanged in the two
consecutive vectors. It is easy to obtain an exponential lower bound -- just encode a binary counter. We start with number $111\ldots11_2$ and count down to $000\ldots00_2$. 
For example:

\vskip 0.3cm

$ \begin{matrix}
  1&1&1&1&0&0&0&0 \\
  1&1&0&0&1&1&0&0 \\
  1&0&1&0&1&0&1&0
 \end{matrix}
 $

\vskip 0.3cm

Because we count down it is easy to see that there is no dominating pair of vectors.
Note that the same implementation of binary counter works when we disallow fixing coordinates.

\vskip 0.3cm

$ \begin{matrix}
  1&2&3&4&0&0&0&0 \\
  1&2&0&0&1&2&0&0 \\
  1&0&1&0&1&0&1&0
 \end{matrix}
 $

\vskip 0.3cm

The main idea behind the double exponential sequence is to encode counters with higher basis. 
However, a problem appears: we cannot decrease a coordinate.
It was irrelevant in case of binary counter - decreasing from $1$ to $0$ is just a reset.

Here induction shows up. When we want to decrease a coordinate in the bigger counter $B$
we reset it and wait until it grows to the appropriate value. 
In the same time a counter $S$ with smaller base is launched.
Roughly speaking, because $S$ is decreasing while coordinates of $B$ are growing up
there is no dominating pair during a growing period.

\section{Overview}

\def\Lc{L^\circ}

Let $\Lc_d$ be the maximum length of a non-dominating {\bf cyclic} sequence for the given $d$;
obviously $L_d \geq \Lc_d$. For example, the following sequence of four vectors shows that $\Lc_2 \geq 4$:
\[
(1,1), (0,2), (1,0), (0,0)
\]

We will show that $\Lc_{d+2} \geq \Lc_d (2\Lc_d+1)$. This is done by looping the
$d$ coordinates, and adding two new coordinates in a way which makes the
whole long sequence non-dominating.
These two new coordinates are meant to implement something like a counter of base $2n-1$.
By repeating this several times  we get a double exponential
lower bound for $\Lc_d$, and also for $L_d$.

As an additional verification, we have also implemented our construction in C++.
Source code available at:

\begin{center}
\verb|http://www.mimuw.edu.pl/~erykk/papers/vecseq.cpp|
\end{center}

\section{Construction}

\begin{theorem} $\Lc_{d+2} \geq \Lc_d (2\Lc_d+1)$ 
\label{step}
\end{theorem}

\begin{proof}
Let $(u_0, \ldots, u_{n-1})$ be a non-dominating cyclic sequence of length $n$ and dimension
$d$.

We construct a sequence $(v_0, \ldots, v_{m-1})$ of length $m=n(2n+1)$ and dimension $d+2$.
For convenience, we name the two new coordinates $\vX$ and $\vY$; therefore, $\vX_k = v^{d+1}_k$
and $\vY_k = v^{d+2}_k$. The construction is as follows:

\begin{itemize}
\item (1) $v_{k}^{\ell} = u_{k \bmod n}^{\ell} $ for $\ell \leq d, 0 \leq k < m$
\item (2) $\vX_{2ni+j} = \max(j-i, 0)$ for $0 \leq j < 2n, 0 \leq 2ni+j < m$
\item (3) $\vY_{k} = \vX_{(k+n) \bmod m}$ for $0 \leq k < m$
\end{itemize}

We will show that this sequence is indeed a non-dominating cyclic sequence.

The following table shows the evolution of the two new coordinates.

\vskip 0.3cm

\[
\begin{array}{c|ccccc|cccc}
j & 0 & 1 & 2 & \ldots & n-1 & n & n+1 & \ldots & 2n-1 \\
\hline 
\vX_j & 0 & 1 & 2 & \ldots & n-1 & n & n+1 & \ldots & 2n-1 \\
\vY_j & n & n+1 & n+2 & \ldots & 2n-1 & 0 & 0 & \ldots & n-2 \\
\hline 
\vX_{2n+j} & 0 & 0 & 1 & \ldots & n-2 & n-1 & n & \ldots & 2n-2 \\
\vY_{2n+j} & n-1 & n & n+1 & \ldots & 2n-2 & 0 & 0 & \ldots & n-3 \\
\hline 
\vX_{4n+j} & 0 & 0 & 0 & \ldots & n-3 & n-2 & n-1 & \ldots & 2n-3 \\
\vY_{4n+j} & n-2 & n-1 & n & \ldots & 2n-3 & 0 & 0 & \ldots & n-4 \\
\hline 
\vdots & \vdots & \vdots & \vdots & \ddots & \vdots & \vdots & \vdots & \ddots & \vdots \\
\hline 
\vX_{2n(n-1)+j} & 0 & 0 & 0 & \ldots & 0 & 1 & 2 & \ldots & n \\
\vY_{2n(n-1)+j} & 1 & 2 & 3 & \ldots & n & 0 & 0 & \ldots & 0 \\
\hline 
\vX_{2n^2+j} & 0 & 0 & 0 & \ldots & 0 &&&& \\
\vY_{2n^2+j} & 0 & 1 & 2 & \ldots & n-1 &&&&
\end{array}
\]

\vskip 0.3cm

We will start by showing that the sequence $v^\ell$ is cyclic for each $\ell$.
For $\ell \leq d$, the sequence $v^\ell$ is simply $u^\ell$ repeated $2n+1$ times.
Since $u^\ell$ is cyclic, $v^\ell$ is cyclic too. We also need to check the coordinates
$\vX$ and $\vY$.

For the coordinate $\vX$, we need to verify three cases:

\begin{itemize}
\item $\vX_{2ni+j} \vnext \vX_{2ni+j+1}$, where $2ni+j+1 < m$, follows immediately
from the formula (2) for $j \neq 2n-1$.

\item $\vX_{m-1} \vnext \vX_0$ because $\vX_{0} = 0$.

\item $\vX_{2ni+(2n-1)} \vnext \vX_{2n(i+1)}$ because
$\vX_{2n(i+1)} = \max(0 - (i+1), 0) = 0$.
\end{itemize}

The sequence $\vY$ is a cyclic shift of $\vX$. We already know that
$\vX$ is cyclic, so $\vY$ is cyclic too.

Now, we need to show that our cyclic sequence is non-dominating:
whenever $a < b$, we have $v_a \nvless v_b$. 
If $a \bmod n < b \bmod n$, we know that
$u_{a \bmod n} \nvless u_{b \bmod n}$. Since $u_{a \bmod n}$ is a part of $v_a$, and
$u_{b \bmod n}$ is a part of $v_b$, we get that $v_a \nvless v_b$. 

Now, suppose that $a \bmod n \geq b \bmod n$.
Let $a = 2ni_a + j_a$, and $b = 2ni_b + j_b$, where $0 \leq j_a, j_b < 2n$.
Since $j_a$ can be less than $n$ or not, and $j_b$ can be less than $n$ or not, 
there are four cases, in each one we easily show that $v_a \nvless v_b$.

\begin{itemize}
\item $j_a, j_b <n$, $j_a \geq j_b$, and $i_a < i_b$. In this case we have $\vY_a = n + j_a - i_a
> n + j_b - i_b = \vY_b$. 

\item $j_a, j_b \geq n$, $j_a \geq j_b$, and $i_a < i_b$. In this case we have $\vX_a = j_a - i_a
> j_b - i_b = \vX_b$.

\item $j_a \geq n$, $j_b < n$, $j_a - n \geq j_b$, and $i_a < i_b$. In this case we have 
$\vX_a = j_a - i_a > n+(j_a-n) - i_a > n + j_b - i_b \geq \max(n + j_b - i_b, 0) = \vX_b$.

\item $j_a < n$, $j_b \geq n$, $j_a + n \geq j_b$, and $i_a \leq i_b$. In this case we have 
$\vY_a = n + j_a - i_a \geq j_b - i_b > \max((j_b - n) - i_b, 0) = \vY_b$, where the last inequality follows
from the fact that $j_b > i_b$.
\end{itemize}

This completes the proof.
\end{proof}

\begin{theorem} $L_d \geq \Lc_d \geq 2^{3 \cdot 2^{\lfloor d/2 \rfloor-1}-1}$ for $d \geq 2$.
\end{theorem}

\begin{proof}
Obviously $L_d \geq \Lc_d$ and $\Lc_{d+1} \geq \Lc_d$. Therefore, 
it is enough to show the claim for $d=2c$.

For $c=1$ we have $\Lc_{2c} \geq 4$, which satisfies the formula.

For $c+1$ we apply Theorem \ref{step} and the induction hypothesis:

\[\Lc_{2c+2} \geq \Lc_{2c} (2\Lc_{2c} + 1) \geq 2 (2^{3 \cdot 2^{c-1}-1})^2 =
2^{3 \cdot 2^{(c+1)-1}-1}\]
\end{proof}

\bibliographystyle{plain}
\bibliography{sequences}

\end{document}